\documentclass[11pt]{article}
\usepackage[margin=1.25in]{geometry}
\usepackage[T1]{fontenc}
\usepackage{newpxtext}
\usepackage{amsmath,amsthm}
\usepackage{newpxmath}
\usepackage{setspace}
\usepackage[authoryear,round]{natbib}
\usepackage[hidelinks,hypertexnames=false]{hyperref}
\hypersetup{pdftitle={A Note on the Identification Step in \textquotedblleft A Semistructural Methodology for Policy Counterfactuals\textquotedblright}, pdfauthor={Henri Ker\"anen}}
\usepackage{microtype}
\newtheorem{proposition}{Proposition}
\newtheorem{lemma}{Lemma}
\theoremstyle{remark}
\newtheorem{remark}{Remark}
\newcommand{\E}{\mathbb{E}}
\newcommand{\R}{\mathbb{R}}

\title{A Note on the Identification Step in\\
``A Semistructural Methodology for Policy Counterfactuals''}
\author{Henri Ker\"anen\thanks{Preliminary draft. Comments welcome: \href{mailto:henri.keranen@helsinki.fi}{henri.keranen@helsinki.fi}.}\\
\small University of Helsinki}
\date{August 2026}

\begin{document}
\maketitle

\begin{abstract}
\noindent The New Keynesian example of \citet{Beraja2023} is not identified at its printed calibration: more than one structure satisfies every condition of its identification step. Two of the paper's six identifying restrictions coincide on equilibrium equations consistent with the printed reduced form, leaving eleven conditions to determine an equation's twelve coefficients. The note derives the rank condition the identification step requires, computable from the reduced form and the restrictions alone. The note also documents and corrects a misprint in the paper's counterfactual display. A one-clause amendment to the paper's Theorem~1 restores its conclusion.
\end{abstract}

\section{Introduction}

\citet{Beraja2023} develops an elegant methodology for constructing counterfactuals with respect to changes in policy rules: it does not require committing to a fully specified model, yet it is not subject to the critique of \citet{Lucas1976}. The method applies to dynamic stochastic models whose equilibria are well approximated by a linear representation. Models that match an economy's reduced-form equilibrium under a benchmark policy rule are observationally equivalent, and the method rests on the insight that many of them, regardless of their microfoundations, also generate an identical counterfactual equilibrium under an alternative one: the paper's principle of counterfactual equivalence. The method proceeds in three steps: estimate, from data generated under the benchmark policy, the recursive law of motion of the equilibrium, the reduced-form model; impose enough restrictions on the equilibrium equations to identify, given the reduced-form model, the remaining coefficients, the structure; and solve for the counterfactual equilibrium of the identified structure under the alternative rule. This note concerns the second step, the identification of the structure.

The original paper illustrates the method in a canonical three-equation New Keynesian model, printed in full in its Section~II: structure, benchmark policy rule, and the reduced form they generate. The note documents that the example is not identified at its printed calibration: a researcher confronted with data generated by the example cannot identify its structure using the paper's restrictions. The result is derived from the paper's own displays. Each line of the structure, an equilibrium equation written as a row of twelve coefficients, must satisfy twelve linear conditions: six of consistency with the reduced form, and the paper's six identifying restrictions. At the printed calibration two of the restrictions coincide exactly on the consistent lines, and eleven conditions do not determine twelve coefficients. Together, the reduced form and the restrictions admit structures that are not counterfactually equivalent; they therefore do not determine the counterfactual.

At the center of the note is a rank condition for the paper's identification step, stated for every reduced form and every set of restrictions the paper's framework allows, and computable from those two objects alone. The condition says exactly when a set of restrictions singles out a line of the structure; when the set does not, the condition counts the free parameters that remain in the line. The requirement it expresses is familiar: identification turns on independent identifying variation, not on the number of restrictions. In fact, the proof of the paper's Theorem~1 already relies on the condition: the uniqueness it asserts holds exactly when the condition does. The theorem's statement assumes only that the restrictions are independent; strengthening that assumption to the condition repairs the theorem in one clause.

The note is intended constructively. The paper's identification step selects a structure from within an observationally equivalent class, and what fails at the printed calibration is the selection. The class and the paper's results about it are not in question: the note relies on them throughout, and the second structure it produces is observationally equivalent to the printed one by the paper's own Lemma~A.1.

Section~\ref{sec:setup} sets out the example and the consistency conditions in the paper's notation; Section~\ref{sec:exhibit} derives from its displays that the example is not identified at the printed calibration. Section~\ref{sec:rank} states the rank condition and shows the failure exact. Section~\ref{sec:misprint} documents and corrects a misprint in the paper's counterfactual display. Section~\ref{sec:disc} returns to the statement of Theorem~1, gives the one-clause repair, and discusses detection and remedies.

\section{Setup and notation}\label{sec:setup}

Notation follows \citet{Beraja2023}, Section~II. The endogenous variables are output, inflation, and the nominal interest rate, $x_t=(y_t,\pi_t,i_t)'$, in log deviations from a zero-inflation steady state; the exogenous shocks $z_t=(b_t,a_t,m_t)'$ --- demand, cost-push, and monetary policy --- follow $z_t=Nz_{t-1}+\epsilon_t$ with $N$ diagonal. With $k=3$ endogenous variables and $s=3$ shocks, each equilibrium condition is a \emph{line} $0 = f\,\E_t x_{t+1} + g\,x_t + h\,x_{t-1} + m\,z_t$, i.e., a row vector $v=(f,g,h,m)\in\R^{3k+s}=\R^{12}$. The \emph{structure} $\xi$ stacks the Euler and NKPC lines,
\begin{equation*}
0=\E_ty_{t+1}+\sigma^{-1}\E_t\pi_{t+1}-y_t-\sigma^{-1}i_t+b_t , \qquad 0=\beta\,\E_t\pi_{t+1}+\kappa\,y_t-\pi_t+a_t ;
\end{equation*}
the \emph{policy} $\Theta$ is the interest-rate rule
\begin{equation*}
0=\theta_\pi \E_t\pi_{t+1}+\theta_y y_t - i_t + \theta_i i_{t-1} + m_t .
\end{equation*}
The paper specifies the example by printing its structure and benchmark policy numerically, each row a line in the block order $\E_tx_{t+1}$, $x_t$, $x_{t-1}$, $z_t$:
\begin{equation}\label{eq:xi0}
\begin{aligned}
\xi^0&=\begin{bmatrix} 1 & \tfrac12 & 0 & -1 & 0 & -\tfrac12 & 0 & 0 & 0 & 1 & 0 & 0\\
0 & 0.9 & 0 & 0.3 & -1 & 0 & 0 & 0 & 0 & 0 & 1 & 0
\end{bmatrix},\\
\Theta^0&=\begin{bmatrix} 0 & 0.7 & 0 & 0.5 & 0 & -1 & 0 & 0 & 0.6 & 0 & 0 & 1
\end{bmatrix}
\end{aligned}
\end{equation}
--- that is, $\sigma=2$, $\beta=0.9$, $\kappa=0.3$, and $(\theta_\pi,\theta_y,\theta_i)=(0.7,\,0.5,\,0.6)$.\footnote{The displayed $\Theta^0$ has $\theta_y=0.5$, while the paper's text introduces the counterfactual as ``$\theta_y=0$ instead of $\theta_y=0.4$''; the printed $P^0,Q^0$ of \eqref{eq:printed} are consistent with $\theta_y=0.5$, and with neither $0.4$ nor $0$. The paper's Section~II.B counterfactual display contains a further, more consequential inconsistency; see Section~\ref{sec:misprint}.}

Under determinacy (the paper's Assumptions~1--2) the structural model generates the reduced form $x_t=P^0x_{t-1}+Q^0z_t$; the paper prints the \emph{reduced-form model} $\Gamma^0\equiv\{P^0,Q^0,N^0\}$ as
\begin{equation}\label{eq:printed}
P^0=\begin{bmatrix} 0 & 0 & -0.35\\
0 & 0 & -0.16\\
0 & 0 & \phantom{-}0.38
\end{bmatrix},
\qquad Q^0=\begin{bmatrix} 1.70 & -0.50 & -0.59\\
1.46 & \phantom{-}3.24 & -0.27\\
1.59 & \phantom{-}1.61 & \phantom{-}0.63
\end{bmatrix},
\qquad N^0=\begin{bmatrix} 0.9 & 0 & 0\\
0 & 0.9 & 0\\
0 & 0 & 0
\end{bmatrix} .
\end{equation}
Note the last diagonal entry of $N^0$: \emph{the monetary policy shock is i.i.d.} The exhibit of Section~\ref{sec:exhibit} turns on these displays.

By the method of undetermined coefficients \citep{Uhlig1995}, a line $v=(f,g,h,m)$ is \emph{consistent} with $\Gamma^0$ --- it holds identically in the state $(x_{t-1},z_t)$, hence on every equilibrium path --- if and only if
\begin{equation}\label{eq:brackets}
fP^2+gP+h=0 \qquad\text{and}\qquad f(QN+PQ)+gQ+m=0 ,
\end{equation}
i.e., $vK=0$ where
\begin{equation*}
K \;\equiv\;
\begin{bmatrix}
P^2 & QN+PQ\\
P & Q\\
I_k & 0\\
0 & I_s
\end{bmatrix}
=\begin{bmatrix}\Psi\\ I_{k+s}\end{bmatrix} \;\in\;\R^{12\times 6} , \qquad \Psi\equiv
\begin{bmatrix}
P^2 & QN+PQ\\
P & Q
\end{bmatrix} .
\end{equation*}
The block rows of $K$ express the four arguments of a line in terms of the state $(x_{t-1},z_t)$: the identity block is the state itself, and the upper block $\Psi$ collects the responses of $\E_tx_{t+1}$ and $x_t$ to it. The matrix $K$ is the paper's own object: it appears, unnamed, as the right-hand factor of its display (SMERLM-NK),\footnote{As printed, the third block row of the display's right-hand factor reads $0_{3,3}$ where $I_3$ belongs; with the zero block, the $x_{t-1}$ coefficients --- here the rule's $\theta_i$ --- would drop from the system altogether. The paper's general form of the display and the online appendix's (Null~OE) place the identity correctly; the display is read here as intended.} which determines $\{P^0,Q^0\}$, the stable solution, by requiring all three lines of the structural model to satisfy \eqref{eq:brackets} --- in matrix form, $\bigl[\begin{smallmatrix}\xi^0\\ \Theta^0\end{smallmatrix}\bigr]K=0$. Consistency reads the same equation in the opposite direction: the reduced form is held fixed, and $vK=0$ asks which lines it supports --- a linear question, whereas solving for $\{P^0,Q^0\}$ was a nonlinear one. The paper's online appendix makes precisely this reversal: its Lemma~A.1 states that a structure is observationally equivalent to $\xi^0$ under the benchmark policy if and only if each of its lines solves this system; the coefficient matrix of its display (Null~OE) is exactly $K'$.

The paper's identification step stacks \eqref{eq:brackets} with its six per-line restrictions, which demand that the structure be of the form
\begin{equation}\tag{Restrictions}\label{eq:restr}
\xi=\begin{bmatrix} \xi_{11} & \xi_{12} & 0 & \xi_{14} & 0 & \xi_{16} & \xi_{17} & \xi_{18} & 0 & 1 & 0 & 0\\
\xi_{21} & \xi_{22} & 0 & \xi_{24} & \xi_{25} & 0 & \xi_{27} & \xi_{28} & 0 & 0 & 1 & 0
\end{bmatrix}
\end{equation}
in the coordinates of \eqref{eq:xi0}; the entries written $\xi_{lj}$ are unrestricted, and the rows are the Euler and NKPC patterns.\footnote{As printed, the (2,1) entry of the paper's display reads $\xi_{12}$ --- the same symbol as its (1,2) entry --- where $\xi_{21}$ belongs; the paper's accompanying text indexes the NKPC row's coefficients $\xi_{2j}$ (it names $\xi_{23}$, $\xi_{26}$, $\xi_{29}$). The display is printed here as intended.} The paper states the restrictions in economic terms: ``only demand shocks shift the Euler equation'', for example (the $a_t$- and $m_t$-coefficients of the Euler row set to zero), or ``the interest rate does not appear in the NKPC'' (the second row's three interest-rate coefficients, on $\E_ti_{t+1}$, $i_t$, and $i_{t-1}$, set to zero).

\section{The example is not identified}\label{sec:exhibit}

Each line of the structure must satisfy twelve conditions: the six of \eqref{eq:brackets} --- consistency with the reduced-form model $\Gamma^0$ generated under the benchmark policy --- and the six of \eqref{eq:restr}. The paper's Theorem~1 concludes that imposing the six restrictions is enough: each line is the unique solution of these twelve conditions. That the conclusion fails at the printed calibration can be established by counting the conditions themselves: for a line consistent with $\Gamma^0$, two of its six restrictions coincide, and eleven conditions do not determine twelve coefficients. An exclusion supplies identifying variation only if the excluded variable moves the equilibrium in a direction of its own.

Both patterns of \eqref{eq:restr} exclude the two variables that belong to the policy rule alone: the lagged interest rate and the monetary policy shock. The shock, i.i.d., shifts the rule once and carries no news about future policy; the lagged rate too shifts it once, with the smoothing coefficient $\theta_i=0.6$. A line $v$ of the structure meets the pair in two ways: directly, through its own coefficients on them, $v_{i_{t-1}}$ and $v_{m_t}$ (the entries the two exclusions target), and indirectly, through its other arguments, since $\E_tx_{t+1}$ and $x_t$ move when the pair moves. Consistency with $\Gamma^0$ demands that the line hold identically in the state $(x_{t-1},z_t)$, so, variable by variable, the direct coefficient must cancel the indirect appearance. The indirect appearances can be written out. The lagged rate moves $x_t$ by $p_i$, the $i_{t-1}$-column of $P^0$; the shock moves it by $q_m$, the $m_t$-column of $Q^0$ --- the reduced-form responses of $x_t$ to the lagged rate and to the shock. Advancing the reduced form one period and taking expectations, $\E_tx_{t+1}=P^0x_t+Q^0N^0z_t$: the lagged rate moves the expectation by $P^0p_i$, the shock by $P^0q_m$ --- plus a forecast term, $q_m$ scaled by the shock's own persistence, zero for the i.i.d.\ shock. These responses enter the line through its blocks $f$ and $g$: its indirect appearance at the lagged rate is $f(P^0p_i)+g\,p_i=(fP^0+g)\,p_i$; at the shock, $(fP^0+g)\,q_m$. One row, $fP^0+g$ (what the line picks up through $\E_tx_{t+1}$ and $x_t$), multiplies both responses. The $i_{t-1}$-column of the first condition in \eqref{eq:brackets} and the $m_t$-column of the second therefore read
\begin{equation*}
(fP^0+g)\,p_i+v_{i_{t-1}}=0, \qquad (fP^0+g)\,q_m+v_{m_t}=0.
\end{equation*}
Each condition solves for the targeted entry: $v_{i_{t-1}}$ and $v_{m_t}$ are determined linearly by $f$ and $g$ --- the blocks behind the indirect appearances.

The three lines of the structural model are consistent by construction (that is how $\{P^0,Q^0\}$ was determined), and their direct coefficients are printed in \eqref{eq:xi0}: zero in the Euler and NKPC lines, $\theta_i$ and $1$ in the rule. Stack the three rows $fP^0+g$ into a matrix $J$; across the three lines, the two conditions read as two systems in the columns:
\begin{equation*}
J\,p_i=-\begin{pmatrix}0\\0\\\theta_i\end{pmatrix}, \qquad J\,q_m=-\begin{pmatrix}0\\0\\1\end{pmatrix}.
\end{equation*}
The two right-hand sides are proportional --- the first is $\theta_i$ times the second --- and the matrix $J$ is invertible under determinacy, so the solutions inherit the proportion: $p_i=\theta_i\,q_m$, an identity in the example's parameters.

The remaining columns of $P^0$ are zero. At $y_{t-1}$ and $\pi_{t-1}$ the first condition in \eqref{eq:brackets} reads as it did at the lagged rate, but with no direct coefficient to do the cancelling: no line of \eqref{eq:xi0} carries lagged output or inflation. The indirect appearances must therefore be zero on their own: across the three lines, $J$ times each of the first two columns of $P^0$ is zero. With $J$ invertible, so are the columns themselves --- $p_i=\theta_i\,q_m$ is the only column $P^0$ has. The printed page displays it: in \eqref{eq:printed} the first two columns of $P^0$ are zero, and, to two decimals, the third is $0.6$ times the $m_t$-column of $Q^0$. Of the lagged state, only the rate reaches the equilibrium, and only through the intercept that the rule feeds the system, $\theta_i\,i_{t-1}+m_t$. To the rest of the economy the pair is thus the same event at different sizes: two states that differ only in how the intercept is split between the lagged rate and the shock generate the same path of output, inflation, and the interest rate.

What the equilibrium bundles, no consistent line can separate. With $p_i=\theta_i\,q_m$, what $v_{i_{t-1}}$ and $v_{m_t}$ cancel is one and the same scalar --- $(fP^0+g)\,q_m$, the line's exposure to the intercept --- once at scale $\theta_i$ and once at scale $1$: every line consistent with $\Gamma^0$ carries $v_{i_{t-1}}=\theta_i\,v_{m_t}$. An exclusion, imposed on a consistent line, therefore strikes no free unknown: consistency has already solved the targeted entry out, and the exclusion is one linear restriction on the coefficients that remain free. The pair supplies a single direction of identifying variation to restrict against --- two excluded instruments with collinear first stages count as one. A consistent line thus excludes $i_{t-1}$ if and only if it excludes $m_t$ --- one restriction stated twice.

The count is now short. Given consistency, the two exclusions have become one restriction. The paper's six restrictions thus amount to five, and a line's twelve conditions to eleven: eleven linear equations in twelve unknowns. The system is solvable: the printed line satisfies all twelve conditions, the restrictions visibly in \eqref{eq:xi0}, and it is consistent by construction. Through each of $\xi^0$'s lines therefore runs at least a one-parameter family of solutions of the same twelve conditions.

The paper's Theorem~1 concludes that exactly one structural model satisfies its conditions (i)--(ii): reduced form $\Gamma^0$ under the policy $\Theta^0$, and the restrictions in each line of the structure. The family delivers a second. Substitute for the printed NKPC line any other member $\bar v$ of its family, keeping the Euler line as it stands in \eqref{eq:xi0}; the new structure then differs from $\xi^0$ in its NKPC line alone. None of the conditions ties one line to the other: \eqref{eq:brackets} constrains one line at a time, and each pattern of \eqref{eq:restr} binds its own line, so each line of the new structure can be checked alone for the same two properties, consistency with $\Gamma^0$ and its own restrictions. The new line $\bar v$ brings both by membership: the twelve conditions defining its family are exactly the six of \eqref{eq:brackets} and the six of its pattern in \eqref{eq:restr}. The Euler line, untouched, brings both from the printed page: its restrictions are visible in \eqref{eq:xi0}, and it is consistent by construction. Every line of the new structure is thus consistent with $\Gamma^0$, and by the paper's own Lemma~A.1 the structure is observationally equivalent to $\xi^0$ under the benchmark policy. Paired with the same $\Theta^0$, it is a second structural model satisfying the conditions (i)--(ii) of the paper's Theorem~1 as stated. The example is not identified.

\section{The rank condition}\label{sec:rank}

Identification of a line amounts to eliminating every $\Gamma^0$-consistent line other than the true one, so the first step is to describe them all.

\begin{lemma}[Consistent lines]\label{lem:null}
A line $v=(f,g,h,m)$ is consistent with a reduced-form model $\Gamma=\{P,Q,N\}$ if and only if
\begin{equation}\label{eq:pinning}
(h,\,m)\;=\;-(f,\,g)\,\Psi ,
\end{equation}
so the set of $\Gamma$-consistent lines is parameterized freely by $(f,g)$: a subspace of dimension exactly $2k$, with basis the rows of
\begin{equation*}
B\equiv\bigl[\,I_{2k} \mid -\Psi\,\bigr] =\begin{bmatrix} I_k & 0 & -P^2 & -(QN+PQ)\\
0 & I_k & -P & -Q
\end{bmatrix}
\in\R^{2k\times(3k+s)} .
\end{equation*}
\end{lemma}

\begin{proof}
By the block form of $K$, $vK=(f,g)\Psi+(h,m)$, so $vK=0$ is exactly \eqref{eq:pinning}: the $\Gamma$-consistent lines form the graph of $(f,g)\mapsto-(f,g)\Psi$, that is, the row space of $B$, whose $I_{2k}$ block makes the rows independent.
\end{proof}

Written out as equations, the rows of $B$ are identities that hold on every equilibrium path regardless of the behavior that generated them: with $e_j$ the $j$-th coordinate vector, row $k+j$ is the law of motion of the $j$-th variable, $0=x_{j,t}-e_j'Px_{t-1}-e_j'Qz_t$, and row $j$ is its expectation one period ahead with $x_t$ substituted out, $0=\E_tx_{j,t+1}-e_j'P^2x_{t-1}-e_j'(QN+PQ)z_t$. Adding any combination of them to an equilibrium equation produces another equation that fits the data equally well; identification requires the restrictions to eliminate this $2k$-dimensional span.\footnote{The ingredients are in the paper's online appendix: its Lemma~A.1 shows that every line of a structure observationally equivalent to $\xi^0$ under the benchmark policy solves $vK=0$; the comment that follows it notes the system is underdetermined --- $3k+s$ unknowns against $k+s$ equations; and the proof of its Proposition~A.1 already recovers the rest of a line from its $(f,g)$-half. Lemma~\ref{lem:null} adds the exact dimension and the explicit basis.}

Nothing in the lemma is specific to the example or its dimensions. In the paper's general model (SME), the structure's conditions carry one additional term, in $\E_tz_{t+1}$; the paper's Definition~1 accordingly collects a structure's shock coefficients into the single block $(LN+M)$, so a line of a general structure is again a vector of $3k+s$ coordinates $v=(f,g,h,m)$, with $m=(LN+M)_l$. The policy lines of (SME) carry no such term at all.\footnote{The blocks $(f,g,h,m)$ lowercase Definition~1's $\xi\equiv[F\;G\;H\;(LN+M)]$; the Section~II displays name no blocks. In the paper's notation the structure has $k-p$ lines, $p$ of the $k$ variables being policy variables.} Either way, a line's consistency with a reduced-form model $\Gamma=\{P,Q,N\}$ reads exactly as in \eqref{eq:brackets}, and the lemma applies unchanged.

Section~\ref{sec:exhibit}'s argument was a count, and counting goes only so far: it detects a failure without delimiting it, and it bounds a solution set's dimension only from below. Two questions are left open: which menus of restrictions stay independent on the consistent lines, and whether a family the count finds is the whole solution set of a line's conditions. Both turn on the rank of one matrix, computable from the reduced form and the menu alone. The matrix is read off $B$. The paper's Theorem~1 takes as given, for each line of the structure, a menu of $2k$ linear restrictions (its $\{R_l,r_l\}$) and assumes them independent.

\begin{proposition}[The rank condition]\label{prop:rank}
Let $\Gamma$ be a reduced-form model and $Rv'=r$ a menu of $2k$ linear restrictions on a line, with $R=[\,R_{fg}\mid R_{hm}\,]\in\R^{2k\times(3k+s)}$ split along the line's blocks. The menu is satisfied by exactly one $\Gamma$-consistent line for every $r$ if and only if
\begin{equation}\label{eq:rank}
RB'\;=\;R_{fg}-R_{hm}\,\Psi'\;\in\;\R^{2k\times2k}
\end{equation}
is nonsingular. When $RB'$ is singular, no value of $r$ delivers uniqueness: the $\Gamma$-consistent lines satisfying the menu, if any, form an affine family of dimension $2k-\operatorname{rank}(RB')$.
\end{proposition}

\begin{proof}
On the consistent lines $v=aB$ of Lemma~\ref{lem:null}, $a=(f,g)$ free, the menu reads $RB'\,a'=r$, a square linear system in the $2k$ free coordinates; since the rows of $B$ are a basis, the system's solutions correspond one-to-one and linearly to the consistent lines satisfying the menu. It has a unique solution for every $r$ if and only if $RB'$ is nonsingular, that is, if and only if the menu eliminates the span of the rows of $B$; otherwise its solution set, and with it the family of such lines, is empty or affine of dimension $2k-\operatorname{rank}(RB')$.
\end{proof}

Independence of the menu is not assumed: dependent rows make \eqref{eq:rank} singular outright, so nonsingularity already entails it. The criterion is, moreover, the paper's own determinant. The proof of Theorem~1 stacks the menu beneath the $k+s$ conditions of \eqref{eq:brackets} into a $(3k+s)$-square system, the system the paper displays as ``exactly determined''. Because the consistency rows carry an identity block in the $(h,m)$-columns, they eliminate those coefficients, the same substitution Lemma~\ref{lem:null} already made, and the elimination reduces the stacked system to \eqref{eq:rank}: the system's determinant equals $\det(RB')$ exactly,\footnote{Post-multiplying the stacked matrix by the unimodular $\bigl[\begin{smallmatrix}I_{2k}&0\\-\Psi'&I\end{smallmatrix}\bigr]$ sends it to $\bigl[\begin{smallmatrix}0&I\\RB'&R_{hm}\end{smallmatrix}\bigr]$; the block swap costs $(-1)^{2k(k+s)}$, an even power, so the equality holds sign included.} and its rank equals $(k+s)+\operatorname{rank}(RB')$. Singularity of \eqref{eq:rank} is thus singularity of the system the paper displays, whose entire rank deficiency sits in the menu's $2k$ rows.

In the example the matrix is singular, and Section~\ref{sec:exhibit}'s identity is the reason: it makes two rows of \eqref{eq:rank} collinear. An exclusion of $i_{t-1}$ or of $m_t$ is a menu row with zero $(f,g)$-part and a single unit entry in the $(h,m)$-blocks, so its row of $RB'$ is, up to sign, a transposed column of $\Psi$, the responses of $\E_tx_{t+1}$ and $x_t$ to the excluded variable. For the lagged rate that column is $(P^0p_i,\;p_i)$; for the shock it is $(P^0q_m,\;q_m)$, where the $N^0$-term has once more vanished because the shock is i.i.d. With $p_i=\theta_iq_m$, the first is $\theta_i$ times the second: the two exclusions restrict against one and the same direction of identifying variation. $R_{fg}-R_{hm}\Psi'$ is therefore singular for every menu containing both, whatever its other $2k-2$ rows. Both of the paper's patterns exclude the pair. Each pattern's six rows are distinct coordinate vectors, linearly independent: the hypothesis holds as printed. What fails in the example is not the theorem's stated condition but the one its conclusion needs: independence on the consistent lines.

The same matrix decides whether the collapse goes further. In the example $2k=6$, so identification of a line needs rank six; at the printed calibration $R_{fg}-R_{hm}\Psi'$ has rank exactly five for each of the two patterns. From above, the identity caps the rank: with two of the six rows collinear, they span at most five directions. From below, a $5\times5$ minor of each pattern's matrix is nonzero: five of its rows are independent, and no further dependence hides. Each line's twelve conditions thus have rank eleven, and by Proposition~\ref{prop:rank} their solution set, which contains the printed line, is an affine family of dimension exactly one: Section~\ref{sec:exhibit}'s ``at least'' is exact, and the families are not part of the solution sets but all of them.

For the NKPC pattern the family's direction has a closed form in the example's parameters. Write $\rho=0.9$ for the common persistence of the demand and cost-push shocks in $N^0$, and $e_y$, $e_\pi$, $e_i$ for the coordinate vectors of output, inflation, and the rate, and take
\begin{equation}\label{eq:dir}
v\;=\;(\,f,\;-\rho f,\;0,\;0\,), \qquad f\;\equiv\;\bigl(\,e_\pi'q_m,\;-\,e_y'q_m,\;0\,\bigr),
\end{equation}
a nonzero combination of the output and inflation coordinates that is orthogonal to $q_m$. Section~\ref{sec:exhibit} left $P^0$ with a single nonzero column, $P^0=\theta_i\,q_m e_i'$, so $fP^0=\theta_i\,(fq_m)\,e_i'=0$, and both brackets of \eqref{eq:brackets} close by short algebra: $v$ is a consistent line, zero in every coefficient the NKPC pattern restricts. The solution set through the printed line consists of its translates by $\lambda v$, $\lambda\in\R$. Written as an equation, $v$ says $\E_t[fx_{t+1}]=\rho\,fx_t$ --- an identity of the reduced form, true regardless of the behavior that generated it. For the Euler pattern no equally short expression offers itself, but by the same rank count its direction exists and is unique up to scale, and members of either family differ from the printed lines only in coefficients the paper's restrictions never touch.

\begin{remark}[The failure is consequential]\label{rem:conseq}
The multiplicity is not a harmless normalization: members of the families are not counterfactually equivalent. Scale each family's direction to a unit coefficient on its line's own expectation term --- $\E_ty_{t+1}$ for the Euler pattern, $\E_t\pi_{t+1}$ for the NKPC --- and add half of it to that line of $\xi^0$. The resulting structure reproduces $\Gamma^0$ and satisfies \eqref{eq:restr} in both of its lines, yet under the paper's counterfactual rule ($\theta_y=0$) it puts the impact response of output to a demand shock at $4.32$ against the true $3.51$.\footnote{The paper's printed counterfactual display shows $5.32$ for this entry; the display is itself misprinted (it solves the shock-loading step at the benchmark $P^0$), and $3.51$ is the corrected value. See Section~\ref{sec:misprint}. Here and in what follows, every decimal not printed in the paper is computed from the exact solution.} A finer reading of condition~(i) of the paper's Theorem~1 does not remove the multiplicity: members near the printed lines yield determinate structures that generate $\Gamma^0$ as their unique reduced form. Finally, \eqref{eq:brackets} and \eqref{eq:rank} contain only $\Gamma^0$ and the menu: the observed policy $\Theta^0$ enters nowhere. Section~\ref{sec:disc} returns to this.
\end{remark}

\section{A misprint in the counterfactual display}\label{sec:misprint}

Comparing the numbers of Remark~\ref{rem:conseq} with the paper requires one correction to its Section~II.B display. That section considers a more hawkish rule, the counterfactual $\Theta^1$, which sets $\theta_y=0$ and leaves every other coefficient of $\Theta^0$ unchanged. It displays two counterfactual reduced forms, the paper's $\Gamma^1$ and $\tilde\Gamma^1$: $\{P^1,Q^1\}$ for the structural model $\{\xi^0,\Theta^1\}$, and $\{\tilde P^1,\tilde Q^1\}$ for $\{\xi^1,\Theta^1\}$, where $\xi^1$ is the working-capital structure that the paper's Section~II.A exhibits as observationally equivalent to $\xi^0$ under the benchmark policy. The printed transition matrices $P^1$ and $\tilde P^1$ are correct at display accuracy (within a unit in the last decimal).\footnote{So is the printed structure $\xi^1$: its Euler line is that of \eqref{eq:xi0}, and a line with the pattern of the paper's working-capital Phillips curve --- its display (NKPC2), after \citet{Christiano2010} --- and with the printed $a_t$-coefficient $1.5$ is consistent with $\Gamma^0$ at coefficients $(0.7001,\,0.3826,\,0.0449)$ on $\E_t\pi_{t+1}$, $y_t$, and $i_t$; the printed second row is its two-decimal rounding.} The printed shock loadings $Q^1$ and $\tilde Q^1$, however, are not.

Section~\ref{sec:exhibit}'s identity detects the error in print. Its derivation used only three features of the environment: that the lagged rate and the monetary policy shock belong to the rule alone, with direct coefficients $\theta_i$ and $1$; that the shock is i.i.d.; and that no line carries lagged output or inflation. $\Theta^1$ changes only $\theta_y$, so all three features survive. The counterfactual pair $\{P^1,Q^1\}$ must satisfy the same identity: the $i_{t-1}$-column of $P^1$ equal to $\theta_i$ times the $m_t$-column of $Q^1$, and zeros in the other two columns. The pair $\{\tilde P^1,\tilde Q^1\}$ owes it as well: the shock is unchanged, and the lines of the printed $\xi^1$ carry neither the excluded pair nor lagged output and inflation, so all three features hold for $\{\xi^1,\Theta^1\}$. In the printed pair $\{P^1,Q^1\}$, the zeros are in place but the proportion is not: in the output row, $0.6\times(-0.84)=-0.504$ stands against the printed $-0.63$. In Section~\ref{sec:exhibit}, the benchmark display passed the same test. Any true reduced form satisfies the consistency conditions \eqref{eq:brackets} by construction, and the identity is one of their consequences: in failing it, the printed pair $\{P^1,Q^1\}$ violates \eqref{eq:brackets} itself. The printed transitions are the true counterfactual ones at display accuracy, so the error is confined to the loadings. The proximate cause is traceable: re-solving the second bracket of \eqref{eq:brackets} for each pair's loadings, with the benchmark $P^0$ held in place of the pair's own transition matrix, reproduces the printed $Q^1$ in all nine entries at printed precision, and the printed $\tilde Q^1$ up to rounding in the final digit.

Re-solving the same bracket at the correct transitions instead yields the display as it should read (rows and columns ordered as in \eqref{eq:printed}; transitions as printed, loadings corrected):
\begin{equation*}
\begin{aligned}
P^1&=\begin{bmatrix} 0 & 0 & -0.63\\
0 & 0 & -0.34\\
0 & 0 & \phantom{-}0.48
\end{bmatrix},
&\qquad Q^1&=\begin{bmatrix} 3.51 & -2.16 & -1.05\\
3.06 & \phantom{-}1.02 & -0.56\\
1.56 & \phantom{-}0.52 & \phantom{-}0.81
\end{bmatrix};\\[4pt]
\tilde P^1&=\begin{bmatrix} 0 & 0 & -0.63\\
0 & 0 & -0.33\\
0 & 0 & \phantom{-}0.49
\end{bmatrix},
&\qquad \tilde Q^1&=\begin{bmatrix} 3.64 & -2.49 & -1.05\\
2.98 & \phantom{-}1.17 & -0.56\\
1.52 & \phantom{-}0.60 & \phantom{-}0.81
\end{bmatrix} .
\end{aligned}
\end{equation*}
The proportion is restored: in both pairs, the $i_{t-1}$-column again stands to the $m_t$-column in the ratio $0.6$ to within display rounding. The correction changes quantitative readings: under the more hawkish rule, the impact response of output to a demand shock roughly doubles, from $1.70$ in \eqref{eq:printed} to $3.51$, whereas the printed display would have it roughly triple, from $1.70$ to $5.32$. The qualitative message of the paper's Section~II.B stands: the two counterfactual reduced forms still differ, $\tilde\Gamma^1\neq\Gamma^1$, so the working-capital model illustrates the Lucas critique.

\section{Discussion}\label{sec:disc}

\paragraph{The statement of Theorem 1.}
The theorem hypothesizes ``a set of $2k$ independent linear restrictions on the coefficients in line $l$ of a structure $\xi$'' and concludes uniqueness. The example shows that the hypothesis, as printed, does not secure the conclusion: Section~\ref{sec:rank} verified the hypothesis for both patterns of \eqref{eq:restr}; Section~\ref{sec:exhibit} produced a second model. The difficulty is confined to one word. ``Independent'' asks the restrictions to be independent of one another; uniqueness needs them independent on the consistent lines. By Section~\ref{sec:rank}, the proof's ``exactly determined'' asserts the stronger property; the statement stops at the weaker. The paper's Section~II.C, invoking the theorem for the general case, states the identification step as a count: ``knowledge of $\{\Gamma^0,\Theta^0\}$ imposes only six restrictions per line in the structure, whereas there are 12 unknown structural coefficients per line. Then, imposing six additional restrictions per line identifies the full structure $\xi^0$.'' Section~\ref{sec:exhibit} carries the count one step further: given consistency, the six additional restrictions amount to five. The paper's count shows its system square, twelve conditions on twelve coefficients; it does not show them independent.

\paragraph{The repair.}
One clause suffices: replace ``independent'' with the rank condition of Proposition~\ref{prop:rank}. The strengthened hypothesis asks that each menu's matrix \eqref{eq:rank} be nonsingular, equivalently that the paper's stacked system be. The conclusion then holds entire.

\begin{proposition}[Identification under the rank
condition]\label{prop:repair} Let $\Gamma$ be a reduced-form model, $\Theta$ a policy, and $\{R_l,r_l\}$ a menu of $2k$ linear restrictions for each structural line $l=1,\dots,k-p$, with $R_lB'$ nonsingular; let $\hat\xi$ be the structure whose $l$-th line is its menu's unique $\Gamma$-consistent solution (Proposition~\ref{prop:rank}). Then every structural model $\{\xi,\Theta\}$ that has reduced form $\Gamma$ under $\Theta$ and satisfies $R_l\xi_l'=r_l$ in each line is $\{\hat\xi,\Theta\}$: at most one model satisfies the two conditions. If moreover $\Gamma$ is stable and the policy's lines are $\Gamma$-consistent, then $\{\hat\xi,\Theta\}$ is itself such a model if and only if it satisfies Assumptions~1--2.
\end{proposition}

\begin{proof}
A model with reduced form $\Gamma$ under $\Theta$ has every line of its structure $\Gamma$-consistent, by the paper's Lemma~A.1; if it also satisfies the menus, each of its structural lines is its menu's unique consistent solution, so its structure is $\hat\xi$. For the second claim, every line of $\{\hat\xi,\Theta\}$ is $\Gamma$-consistent, the structural lines by construction and the policy's by hypothesis, so $\Gamma$ solves the model's equilibrium system. If the model satisfies Assumptions~1--2, its stable solution is unique, and $\Gamma$, being stable, is that solution: the model has reduced form $\Gamma$ and satisfies its menus by construction. Conversely, a model with a reduced form satisfies Assumptions~1--2 by the definition of having one.
\end{proof}

The proposition's two conditions are the theorem's (i) and (ii). Both hypotheses of its existence clause are automatic for an observed pair: the observed reduced form is stable, and the observed policy is part of the model that generated it, so its lines are $\Gamma^0$-consistent. Deciding existence is then a finite computation: one linear solve per line assembles the candidate, and a determinacy check on the assembled model settles the question; by the proposition, the check passes exactly when some structure generating $\Gamma^0$ under the observed policy satisfies the restrictions. The counterfactual half of the theorem is untouched by the failure documented here: given a structure and a counterfactual policy, the counterfactual reduced form comes from the same nonlinear step that determined $\{P^0,Q^0\}$, and the failure belongs to the linear step before it.

\paragraph{A canonical knife edge.}
Section~\ref{sec:exhibit}'s derivation used three features of the example and nothing else: the lagged rate and the policy shock enter through the rule alone; the shock is i.i.d.; and no line carries lagged output or inflation. The knife edge is the shock's lack of persistence: give the shock persistence and the structure would generically be identified.\footnote{Generically, because a single exception exists: the persistence the demand and cost-push shocks already share.} But an i.i.d.\ policy shock alongside interest smoothing is standard New Keynesian practice, in which persistence of the instrument is placed deliberately in the rule, as the term $\theta_i\,i_{t-1}$, rather than in the shock. A researcher who adopts the example as a template would inherit its configuration, and with it the collapse.

\paragraph{The observational-equivalence exhibit.}
While persistence in the policy shock would generically rescue identification in the paper's example, it would come at the price of the paper's Section~II.A demonstration of observational equivalence. The exhibit is a pair of models with distinct microfoundations and one reduced form. The behavioral model after \citet{Gabaix2020} has the structure $\xi^0$ in \eqref{eq:xi0}: at the paper's parameter choices its Phillips curve (NKPC1) is the NKPC line of that structure. The working-capital structure differs in its Phillips curve (NKPC2) alone, so the demonstration hinges, by the paper's Lemma~A.1, on that one line being consistent with the reduced form the two models share.\footnote{The paper's Section~II.A accounts for the equivalence differently: ``The common feature of these models is that their structures satisfy exactly six restrictions per line; see the structure in (Restrictions) below. As lemma~A.1 in the appendix shows in the general case, this feature makes them observationally equivalent.'' However, Lemma~A.1 mentions no restrictions; and the working-capital line carries the interest rate, which the structure in \eqref{eq:restr} excludes from its NKPC row.} Consistency is the six conditions of \eqref{eq:brackets}. In the coordinates of \eqref{eq:xi0}, a line with (NKPC2)'s shape is
\begin{equation*}
\bigl(\,0,\;\beta,\;0\;\big|\;\kappa,\;-1,\;\chi\;\big|\; 0,\;0,\;0\;\big|\;0,\;\gamma,\;0\,\bigr):
\end{equation*}
four free coefficients, on expected inflation, output, the rate, and the cost-push shock; the rate's, $\chi$, is the one that makes the curve working-capital. Six conditions on four coefficients in general leave no line at all. However, three of the six come free at the printed calibration. The line's lagged block is zero and the first two columns of $P^0$ in \eqref{eq:printed} are zero, so the first bracket of \eqref{eq:brackets} demands only at the lagged rate, $(fP^0+g)\,p_i=0$. And the policy shock's condition, the shock i.i.d.\ and the line's $m_t$-coefficient zero, is $(fP^0+g)\,q_m=0$: Section~\ref{sec:exhibit}'s identity $p_i=\theta_i\,q_m$ makes it the lagged rate's over again. The three conditions that remain on the four coefficients leave a one-parameter family of consistent lines with (NKPC2)'s shape, and Section~\ref{sec:misprint}'s working-capital line is one of them, shown consistent there. However, persistence would generically leave four of the six in force: the transition matrix stays the printed $P^0$ at every persistence, its first two columns zero, while the shock's condition gains a forecast term the lagged rate's lacks, and the identity's pair comes apart. Four conditions on four coefficients pin a single line.\footnote{The exception is a single persistence, distinct from the other shocks' common value: there the family returns whole, and Section~\ref{sec:misprint}'s working-capital line is consistent once more.} The line is the printed NKPC one, a line of the structure and so still consistent, with $\chi=0$ in \eqref{eq:xi0}: no (NKPC2) line with a nonzero $\chi$ remains, and no working-capital structure shares the reduced form.

\paragraph{Detection and remedies.}
The failure of the identification step can be caught before any counterfactual is attempted. Remark~\ref{rem:conseq} noted that \eqref{eq:brackets} and \eqref{eq:rank} contain only the reduced form and the menu of restrictions; the observed policy enters neither. The test is one $2k\times2k$ determinant, computable before any use is made of the structure: identification of a line needs $\det(R_{fg}-R_{hm}\Psi')\neq0$. At the printed calibration the determinant is zero for both patterns. What restores identification is a set of $2k$ restrictions independent on the consistent lines, and the determinant's two arguments, the menu and the equilibrium, give the two routes to it. A seventh restriction in each line takes the menu route.\footnote{For example, add to the Euler pattern the real-rate link, the line's $\E_t\pi_{t+1}$- and $i_t$-coefficients summing to zero, and to the Phillips-curve pattern the exclusion of $\E_ty_{t+1}$. The matrix of \eqref{eq:rank} is then taller than square, and the criterion is its rank: at the printed calibration each augmented menu reaches rank $2k$, so each restores identification of its line, and the printed structure and the working-capital alternative satisfy the two added restrictions alike.} Persistence in the policy shock takes the equilibrium route: the example's structural lines carry nothing of the shock's process, so the change leaves them untouched and alters only the identifying variation the equilibrium supplies. An i.i.d.\ shock reaches the equilibrium only through the rule's intercept, as the lagged rate does; a persistent one carries news about future policy, a direction of its own, and the pair comes apart, at the price the previous paragraph records for the observational-equivalence exhibit. The two routes differ in who can take them: the menu is the researcher's to strengthen, the shock's process belongs to the economy, and the determinant tells whether the menu route is needed.

\paragraph{Scope.}
One step of the paper's architecture is at issue: the selection of a structure from within the observationally equivalent class. The class stands, and so does the architecture around the step; this note has relied on both throughout, on Lemma~A.1 above all. With the rank condition in place of ``independent'', Theorem~1's hypothesis secures its conclusion.

\end{document}